\documentclass[a4paper,11pt]{article}
\usepackage{times}
\usepackage{graphics}
\usepackage{graphicx}
\usepackage{realA4}
\usepackage{color}
\usepackage{url}
\usepackage{enumerate}
\usepackage{amsmath,amssymb}

\usepackage{amsthm} 
\newtheoremstyle{plain-boldhead}
  {\topsep}
  {\topsep}
  {\itshape}
  {}
  {\bfseries}
  {.}
  { }
  {\thmname{#1}\thmnumber{ #2}\thmnote{ (\bfseries #3)}}
\newtheoremstyle{definition-boldhead}
  {\topsep}
  {\topsep}
  {\normalfont}
  {}
  {\bfseries}
  {.}
  { }
  {\thmname{#1}\thmnumber{ #2}\thmnote{ (\bfseries #3)}}
\theoremstyle{plain-boldhead}
\newtheorem{theorem}{Theorem}

\theoremstyle{definition-boldhead}
\newtheorem{definition}{Definition}

\newcommand{\negspace}{\vspace{-0.5\baselineskip}}

\newcommand{\tup}[1]{%
  \relax\ifmmode
    \langle #1 \rangle%
  \else
    $\langle$#1$\rangle$%
  \fi}
\newcommand{\token}[1]{\textsc{#1}}

\newcommand{\omitit}[1]{}

\newenvironment{description*}%
  {\begin{description}%
    \setlength{\itemsep}{-5pt}%
    \setlength{\parsep}{-5pt}%
    \setlength{\topsep}{-5pt}}%
  {\end{description}}
\newenvironment{enumerate*}%
  {\begin{enumerate}%
    \setlength{\itemsep}{-5pt}%
    \setlength{\parsep}{-5pt}%
    \setlength{\topsep}{-5pt}}%
  {\end{enumerate}}
\newenvironment{itemize*}%
  {\begin{itemize}%
    \setlength{\itemsep}{0pt}%
    \setlength{\parsep}{0pt}%
    \setlength{\topsep}{0pt}}%
  {\end{itemize}}

\newenvironment{tabbing*}%
  {\begin{tabbing}%
    \setlength{\itemsep}{0pt}%
    \setlength{\parsep}{0pt}%
    \setlength{\topsep}{0pt}}%
  {\end{tabbing}}

\usepackage{color}
\usepackage{float}

\begin{document}

\title{\bf Fork Sequential Consistency is Blocking}

\author{Christian Cachin\thanks{IBM Research, Zurich Research
    Laboratory, CH-8803 R\"uschlikon, Switzerland. cca@zurich.ibm.com}
  \and Idit Keidar\thanks{Department of Electrical Engineering,
    Technion, Haifa 32000, Israel. \{idish@ee, shralex@tx\}.technion.ac.il}
  \and Alexander Shraer\footnotemark[2] 
}


\maketitle
\negspace\negspace\negspace
\begin{abstract}\noindent
We consider an untrusted server storing shared data on behalf of clients.
We show that no storage access protocol can on the one hand
preserve sequential consistency and wait-freedom when the server is correct,
and on the other hand always preserve fork sequential consistency.
\end{abstract}

\negspace\negspace
\section{Introduction}

We examine an online collaboration facility providing storage and data
sharing functions for remote clients that do not communicate directly
\cite{subversion, googledocs, wikipedia:dfs, web20}. 
Specifically, we consider a server that implements single-writer multi-reader registers. 
The storage server may be faulty, potentially exhibiting Byzantine
faults~\cite{mazsha02, lkms04, oprrei06, CSS07}.  When the server is
correct, strong liveness, namely \emph{wait-freedom}~\cite{herlih91},
should be guaranteed, as a client editing a document does not want to
be dependent on another client, which could even be in a different
timezone~\cite{web20}. In addition, although read/write operations of
different clients may occur concurrently, consistency of the shared
data should be provided. Specifically, we consider a service that,
when the server is correct, provides \emph{sequential consistency},
which ensures that clients have the same \emph{view} of the order of
read/write operations, which also respects the local order of
operations occurring at each client~\cite{Lamport79}. Sequential
consistency provides clients with a convenient abstraction of a shared
storage space.  It allows for more efficient implementations than
stronger consistency conditions such as
linearizability~\cite{herwin90}, especially when the system is not
synchronized~\cite{AW94}.

In executions where the server is faulty, liveness obviously cannot be
guaranteed. Moreover, with a Byzantine server, ensuring sequential
consistency is also impossible~\cite{CSS07}. Still, it is possible to
guarantee weaker semantics, in particular 
so-called \emph{forking} consistency notions~\cite{lkms04,mazsha02}.
These ensure that whenever the
server causes the views of two clients to differ in a single
operation, the two clients never again see each other's updates after
that. In other words, if an operation appears in the views of two clients,
these views are identical up to this operation.

Originally, \emph{fork-linearizability} was
considered~\cite{lkms04,mazsha02, CSS07}. In this paper, we examine the weaker
\emph{fork sequential consistency} condition, recently introduced by Oprea and
Reiter~\cite{oprrei06}, who showed that this new condition
is sufficient for certain applications. However, to date, no fork-sequentially-consistent storage
protocol has been proposed. In fact, Oprea and Reiter suggested this as a future
research direction~\cite{oprrei06}. Furthermore, Cachin et al.~\cite{CSS07} showed that the stronger notion of
fork-linearizability does not allow for wait-free implementations, but
conjectured that such implementations might be possible with
fork sequential consistency. Surprisingly, we prove here that no
storage access protocol can provide fork sequential consistency at all
times and also be sequentially consistent and wait-free whenever the
server is correct. This generalizes the impossibility result of Cachin
et al.~\cite{CSS07}, and requires a more elaborate proof.

In this paper we require only sequentially consistent
semantics when the server is correct. Though one may also consider
stronger semantics, such as linearizability, for this case, as our
goal is to prove an impossibility result, it suffices to address
 sequential consistency.  Our impossibility result a fortiori rules out the
existence of protocols with stronger semantics as well.

\section{Definitions}
\label{sec:def}

\paragraph{System model.} We consider an asynchronous distributed
system consisting of $n$ clients $C_1,\ldots, C_n$, a server~$S$, and
asynchronous FIFO reliable channels between the clients and $S$ (there
is no direct communication between clients).  The clients and the
server are collectively called \emph{parties}. System components are
modeled as deterministic I/O Automata~\cite{Lynch96}. An automaton has
a state, which changes according to \emph{transitions} that are
triggered by \emph{actions}. A \emph{protocol}~$P$ specifies the
behaviors of all parties. An execution of $P$ is a sequence of
alternating states and actions, such that state transitions occur
according to the specification of system components.

All clients follow the protocol, and any number of clients can fail by
crashing.  The server might be faulty and deviate arbitrarily from the
protocol, exhibiting so-called ``Byzantine'' faults~\cite{peshla80}.
A party that does not fail in an execution is \emph{correct}.  The
protocol emulates a \emph{shared functionality} $F$ to the clients,
defined analogously to shared-memory objects.

\paragraph{Events, operations, and histories.} Clients interact with
the functionality $F$ via \emph{operations} provided by $F$.  As
operations take time, they are represented by two \emph{events}
occurring at the client, an \emph{invocation} and a \emph{response}.
An operation is \emph{complete} if it has a response.  For a sequence
of events $\sigma$, $\textit{complete}(\sigma)$ is the maximal
subsequence of $\sigma$ consisting only of complete operations.

A \emph{history} is a sequence of requests and responses of $F$
occurring in an execution.  An operation $o$ \emph{precedes} another
operation $o'$ in a sequence of events $\sigma$, denoted $o <_\sigma
o'$, whenever $o$ completes before $o'$ is invoked in $\sigma$. Two
operations are \emph{concurrent} if neither one of them precedes the
other.  A sequence of events is \emph{sequential} if it does not
contain concurrent operations.  A sequence of events $\pi$
\emph{preserves the real-time order} of a history $\sigma$ if for
every two operations $o$ and $o'$ in $\pi$, if $o <_\sigma o'$ then
$o<_\pi o'$.  For a sequence of events $\pi$, the subsequence of $\pi$
consisting of events occurring at client $C_i$ is denoted by
$\pi|_{C_i}$.
For a sequential $\pi$, the prefix of $\pi$ ending with operation
$\textit{o}$ is denoted by $\pi^{\textit{o}}$.

An execution is \emph{admissible} if the
following two conditions hold: (1) the sequence of
events at each client consists of alternating invocations and matching
responses, starting with an invocation; and (2) the execution is fair.
\emph{Fairness} means, informally, that the execution does not halt
prematurely when there are still steps to be taken or messages to be
delivered (we refer to the standard literature for a formal definition
of admissibility and fairness~\cite{Lynch96}).

%

\paragraph{Read/write registers.} A functionality $F$ is defined via a
\emph{sequential specification}, which indicates the behavior of $F$
in sequential executions.

The basic functionality we consider is a \emph{read/write
  register}~$X$.  A register stores a value~$v$ from a domain
$\mathcal{X}$ and offers \emph{read} and \emph{write} operations.
Initially, a register holds a special value $\bot \not\in
\mathcal{X}$.  When a client~$C_i$ invokes a read operation, the
register responds with a value $v$, denoted $\textit{read}_i(X) \to
v$.  When $C_i$ invokes a write operation with value $v$, denoted
$\textit{write}_i(X,v)$, the response of $X$ is an acknowledgment,
denoted by \token{ok}.  The sequential specification requires that
each read operation from $X$ return the value written by the most
recent preceding write operation, if there is one, and the initial
value otherwise.  We assume that the values written to every
particular register are unique, i.e., no value is written more than
once.  This can easily be implemented by including the identity of the
writer and a sequence number together with the stored value.

In this paper, we consider \emph{single-writer/multi-reader (SWMR)}
registers, where for every register, only a designated writer may
invoke the write operation, but any client may invoke the read
operation.

\paragraph{Sequential consistency.}
One of the most important consistency conditions for concurrent access
is sequential consistency~\cite{Lamport79}, which preserves the
real-time order only for operations by the same client.  This is in
contrast to linearizability, which must preserve the real-time order
for all operations.

\begin{definition}[Sequential consistency~\cite{Lamport79}] 
  A history $\sigma$ is \emph{sequentially consistent} w.r.t.
  a functionality $F$ if it can be extended (by appending zero or more
  response events) to a history $\sigma'$, and there exists a
  sequential permutation $\pi$ of \textit{complete}($\sigma'$) such
  that:
\begin{enumerate}
\item For every client $C_i$, the sequence $\pi|_{C_i}$ preserves the real-time
  order of $\sigma$; and
\item The operations of $\pi$ satisfy the sequential specification of $F$.
\end{enumerate}
\end{definition}

Intuitively, sequential consistency requires that every operation
takes effect at some point and occurs somewhere in the permutation
$\pi$.  This guarantees that every write operation is eventually
seen by all clients.  In other words, if an operation writes $v$ to a
register $X$, there cannot be an infinite number of subsequent read
operations from register $X$ that return a value written to $X$ prior
to $v$.

\paragraph{Wait-freedom.} 
A shared functionality needs to ensure liveness.  A common requirement
is that clients are able to make progress independently of the actions
or failures of other clients.  A notion that formally captures this
idea is \emph{wait-freedom}~\cite{herlih91}.
\begin{definition}[Wait-free history]
  A history $\sigma$ is wait-free if every operation by a correct
  client in $\sigma$ is complete.
\end{definition}

\paragraph{Fork sequential consistency.} 
The notion of fork sequential consistency~\cite{oprrei06} requires,
informally, that when an operation is observed directly or indirectly
by multiple clients, then the history of events occurring before the
operation is the same at these clients.  For instance, when a client
reads a value written by another client, the reader is assured to be
consistent with the writer up to its write operation.

\begin{definition}[Fork sequential consistency]\label{def:fork}
  A history $\sigma$ is \emph{fork-sequentially-consistent} w.r.t. a
  functionality~$F$ if it can be extended (by appending zero or more
  response events) to a history $\sigma'$, such that for each client
  $C_i$ there exists a subsequence~$\sigma_i$ of
  $\textit{complete}(\sigma')$ and a sequential permutation~$\pi_i$ of
  $\sigma_i$ such that:
  \begin{enumerate*}
  \item All complete operations in $\sigma|_{C_i}$ are contained
    in~$\sigma_i$;
  \item For every client $C_j$, the
   sequence $\pi_i|_{C_j}$ preserves the real-time order of~$\sigma$;
  \item The operations of $\pi_i$ satisfy the sequential specification
    of $F$; and
  \item (\emph{No-join}) For every $\textit{o} \in \pi_i \cap \pi_j$,
    it holds that $\pi_i^{\textit{o}} = \pi_j^{\textit{o}}$.
  \end{enumerate*}
  A permutation $\pi_i$ satisfying these properties is called a
  \emph{view} of~$C_i$.
\end{definition}

Note that a view $\pi_i$ of $C_i$ contains at least all those
operations that either occur at $C_i$ or are apparent from $C_i$'s
interaction with $F$. A fork-sequentially-consistent history in which
some permutation $\pi$ of $\textit{complete}(\sigma')$ is a possible
view of all clients is sequentially consistent.

We are now ready to define a fork-sequentially-consistent storage
service.  It should guarantee sequential consistency and wait-freedom
when the server is correct, and fork sequential consistency otherwise.

\begin{definition}[Wait-free fork-sequentially-consistent Byzantine emulation]
  A protocol $P$ is a wait-free fork-sequentially-consistent Byzantine
  emulation of a functionality $F$ on a Byzantine server $S$ if $P$
  satisfies the following conditions:
  \begin{enumerate*}
  \item If $S$ is correct, the history of every admissible execution
    of $P$ is sequentially consistent w.r.t. $F$ and wait-free; and
  \item The history of every admissible execution of $P$ is fork
    sequentially consistent w.r.t. $F$.
  \end{enumerate*}
\end{definition}

We show next that wait-free fork-sequentially-consistent Byzantine
emulations of SWMR registers are impossible.

\section{Impossibility of Wait-Freedom with Fork Sequential Consistency}
\label{fork-seq}

  
\begin{theorem}\label{thm:wait}
  There is no wait-free fork-sequentially-consistent Byzantine
  emulation of $n \geq 2$ SWMR registers on a Byzantine server $S$.
\end{theorem}

\begin{proof}
  Towards a contradiction assume that there exists such a
  protocol~$P$.  Then in any admissible execution of $P$ with a
  correct server, every operation of a correct client completes.  We
  next construct three executions $\alpha$, $\beta$, and $\gamma$ of
  $P$, shown in Figures~\ref{fig:waiting-alpha}--\ref{fig:waiting-gamma}.  
  All three executions are admissible, since
  clients issue operations sequentially, and
  every message sent between two correct parties is eventually delivered.
  There are two clients $C_1$ and $C_2$, which
  are always correct, and access two SWMR registers $X_1$ and $X_2$.  
  Protocol $P$ describes the asynchronous
  interaction of the clients with $S$; this interaction is depicted in
  the figures only when necessary.
  
  \paragraph{Execution $\alpha$.} 
  In execution $\alpha$, the server is correct.  The execution is
  shown in Figure~\ref{fig:waiting-alpha} and begins with four
  operations by $C_2$: first $C_2$ executes a write operation with
  value $v_1$ to register $X_2$, denoted $w_2^1$, then an operation
  reading register $X_1$, denoted $r_2^1$, then an operation writing
  $v_2$ to $X_2$, denoted $w_2^2$, and finally again a read operation
  of $X_1$, denoted $r_2^2$.  Since $S$ and $C_2$ are correct and $P$
  is wait-free with a correct server, all operations of $C_2$
  eventually complete.
  
   \begin{figure}[htb]
    \begin{center}
    \includegraphics[height=4cm]{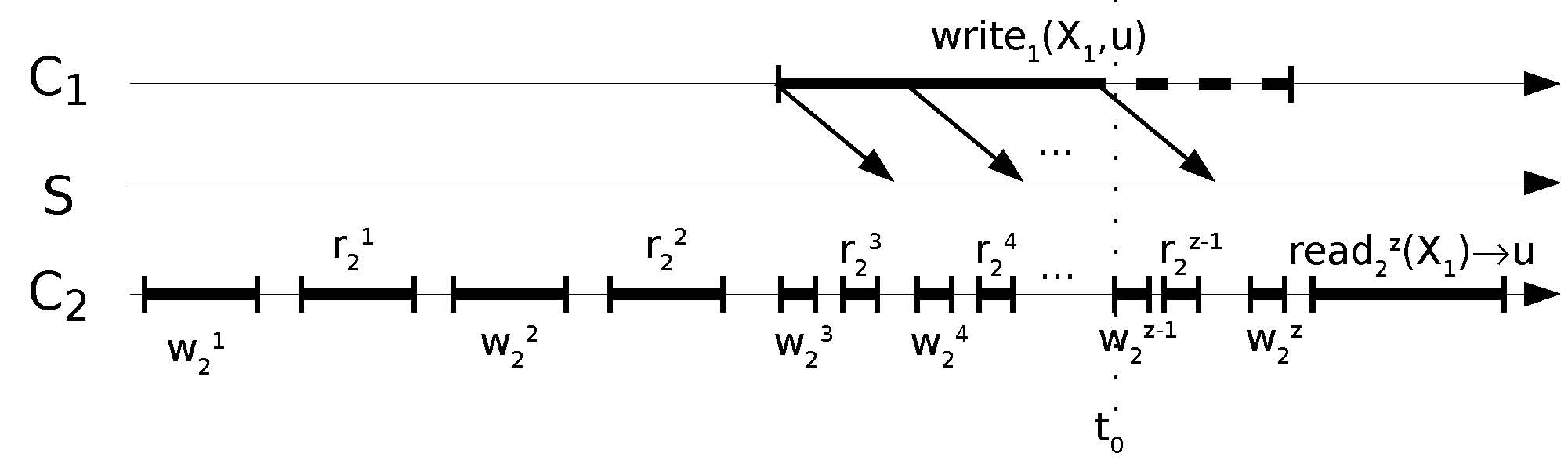}
      \end{center}
     \caption{Execution $\alpha$, where $S$ is correct.}
    \label{fig:waiting-alpha}
  \end{figure}
  
  Execution $\alpha$ continues as follows.  $C_1$ starts to execute a
  single write operation with value $u$ to $X_1$, denoted $w_1$.
  Every time a message is sent from $C_1$ to $S$ during this
  operation, and as long as no read operation by $C_2$ from $X_1$
  returns a value different from $\bot$, the following steps are
  repeated in order, for $i=3, 4, \dots$:
  \begin{enumerate}[(a)]
  \item The message from $C_1$ is delayed by the asynchronous network;
  \item $C_2$ executes an operation writing $v_i$ to $X_2$, denoted
    $w^i_2$;
  \item $C_2$ executes an operation reading $X_1$, denoted $r^i_2$;
    and
  \item the delayed message from $C_1$ is delivered to $S$.
  \end{enumerate}
  Note that $w^i_2$ and $r^i_2$ complete by the assumptions that $P$
  is wait-free and that $S$ is correct. For the same reason, operation
  $w_1$ eventually completes.  After $w_1$ completes, and while $C_2$
  does not read any non-$\bot$ value from $X_1$, $C_2$ continues to
  execute alternating operations $w^i_2$ and $r^i_2$, writing $v_i$ to
  $X_2$ and reading $X_1$, respectively.  This continues until some
  read returns a non-$\bot$ value.  Because $S$ is correct, eventually
  some read of $X_1$ is guaranteed to return $u\neq\bot$ by sequential
  consistency of the execution.  We denote the first such read by
  $r^z_2$.  This is the last operation of $C_2$ in $\alpha$.  If
  messages are sent from $C_1$ to $S$ after the completion of $r^z_2$,
  they are not delayed.
  
  Note that the prefix of $\alpha$ up to the completion of $r^3_2$ is
  indistinguishable to $C_2$ from an execution in which no client
  writes to $X_1$, and therefore $r^1_2$, $r^2_2$, and $r^3_2$ return
  the initial value $\bot$. Hence, $z\geq 4$.

  We denote the point of invocation of $w^{z-1}_2$ in $\alpha$ by
  $t_0$.  It is marked by a dotted line.  Executions $\beta$ and
  $\gamma$ constructed below are identical to $\alpha$ before $t_0$,
  but differ from $\alpha$ starting at $t_0$.
  
  \begin{figure}[htb]
    \begin{center}
      \includegraphics[height=3.6cm]{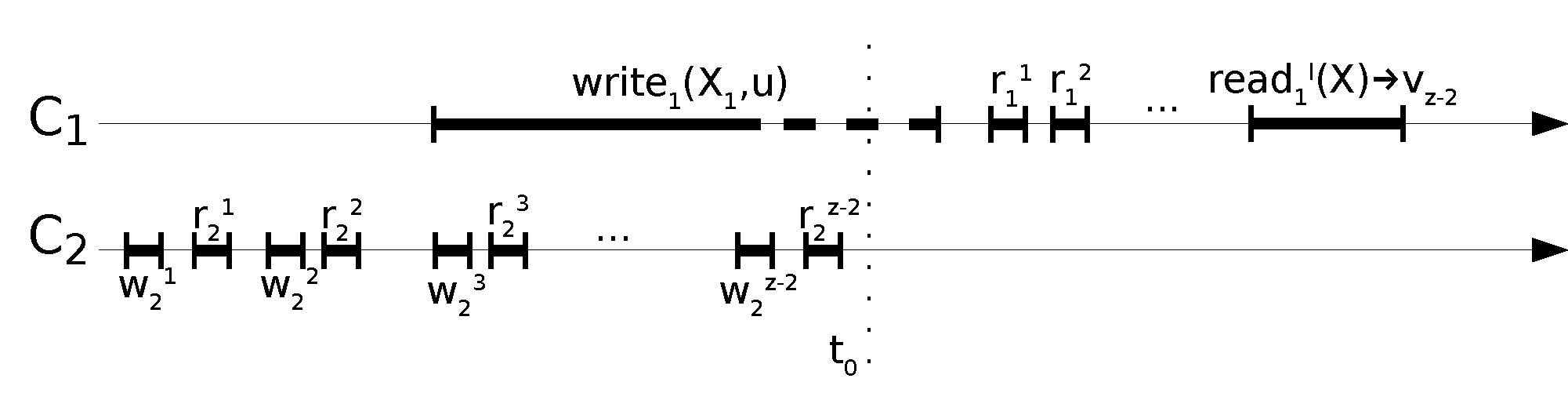}
    \end{center}
    \caption{Execution $\beta$, where $S$ is correct.}
    \label{fig:waiting-beta}
  \end{figure}

  \paragraph{Execution $\beta$.}
  We next define execution $\beta$, shown in
  Figure~\ref{fig:waiting-beta}, in which the server is also correct.
  Execution $\beta$ is identical to $\alpha$ up to the end of
  $r^{z-2}_2$ (before $t_0$), but then $C_2$ halts.  In other words,
  the last two write-read pairs of $C_2$ in $\alpha$ are missing in
  $\beta$.  Operation $w_1$ is invoked in $\beta$ like in $\alpha$ and
  begins after the completion of $r^2_2$ (notice that $r^2_2$ is in
  $\beta$ since $z\geq 4$).  Because the protocol is wait-free with
  the correct server, operation $w_1$ completes.  Afterwards, $C_1$
  repeatedly reads $X_2$ until $v_{z-2}$ is returned.  Because the
  execution is sequentially consistent with the correct server, a read
  of $X_2$ eventually returns~$v_{z-2}$.  We denote the $i$-th read
  operation of $C_1$ by $r^i_1$ and the read operation that returns
  $v_{z-2}$ by $r^l_1$.

  \paragraph{Execution $\gamma$.} 
  The third execution $\gamma$ is shown in
  Figure~\ref{fig:waiting-gamma}; here, the server is faulty.
  Execution~$\gamma$ proceeds just like the common prefix of $\alpha$
  and $\beta$ before $t_0$, and client $C_1$ invokes $w_1$ in the same
  way as in $\alpha$ and in $\beta$.  From $t_0$ onward, the server
  simulates $\beta$ to $C_1$.  This is easy because $S$ simply hides
  from $C_1$ all operations of $C_2$ starting with $w^{z-1}_2$.  The
  server also simulates $\alpha$ to $C_2$.  We next explain how this
  is done.  Notice that in $\alpha$, the server receives at most one
  message from $C_1$ between $t_0$ and the completion of $r^z_2$, and
  this message is sent before $t_0$ by construction of $\alpha$.  If
  such a message exists in $\alpha$, then in $\gamma$, which is
  identical to $\alpha$ before $t_0$, the same message is sent by
  $C_1$.  Therefore, the server has all information needed to simulate
  $\alpha$ to $C_2$ and $r^{z}_2$ returns $u$.

  \begin{figure}[htb]
    \begin{center}
      \includegraphics[height=3.6cm]{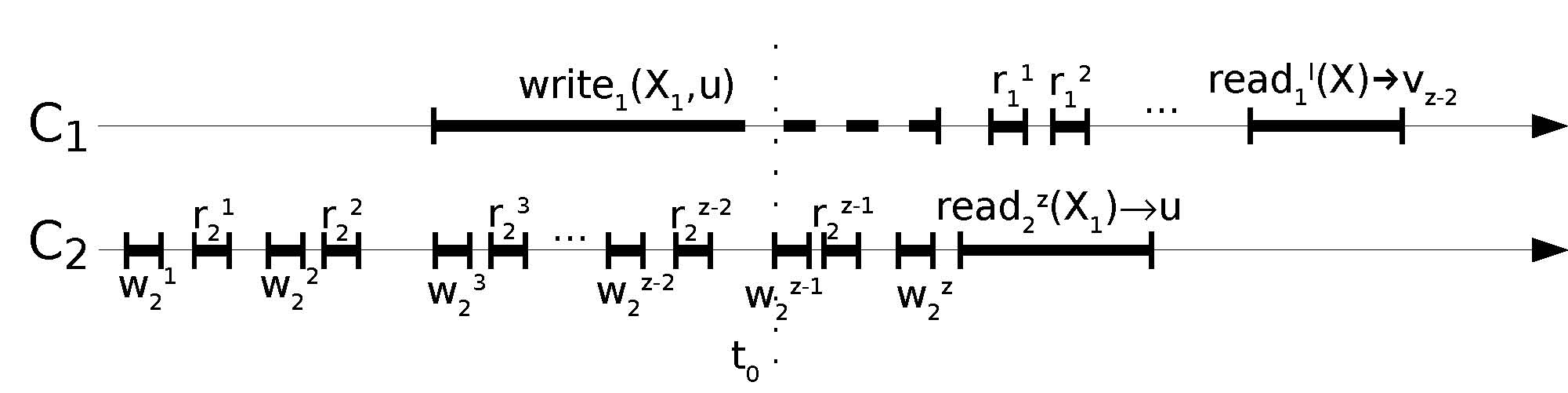}
    \end{center}
    \caption{Execution $\gamma$, where $S$ is faulty and simulates
      $\alpha$ to $C_2$ and $\beta$ to $C_1$.}
    \label{fig:waiting-gamma}
  \end{figure}

  Thus, $\gamma$ is indistinguishable from $\alpha$ to $C_2$ and
  indistinguishable from $\beta$ to $C_1$. However, we next show that
  $\gamma$ is not fork-sequentially-consistent.  Consider the
  sequential permutation~$\pi_2$ required by the definition of fork
  sequential consistency, i.e., the view of $C_2$.  As the real-time
  order of $C_2$'s operations and the sequential specification of the
  registers must be preserved in $\pi_2$, and since $r^1_2$, ...,
  $r^{z-1}_2$ return $\bot$ but $r^z_2$ returns $u$, we conclude that
  $w_1$ must appear in $\pi_2$ and is located after $r^{z-1}_2$ but
  before $r^z_2$.  Because $w_1$ is one of $C_1$'s operations, it also
  appears in $\pi_1$.  By the no-join property, the sequence of
  operations preceding $w_1$ in $\pi_2$ must be the same as the
  sequence preceding $w_1$ in $\pi_1$.  In particular, $w_2^{z-1}$ and $w_2^{z-2}$
  appear in $\pi_1$ before $w_1$, and $w_2^{z-2}$ precedes $w_2^{z-1}$.  Since the real-time order of
  $C_1$'s operations must be preserved in $\pi_1$, operation $w_1$
  and, hence, also $w^{z-1}_2$, appears in $\pi_1$ before $r^l_1$.
  But since $w^{z-1}_2$ writes $v_{z-1}$ to $X_2$ and $r^l_1$ reads 
  $v_{z-2}$ from $X_2$, this violates the sequential specification of $X_2$ 
  ($v_{z-2}$ is written only by $w^{z-2}_2$).  
  This contradicts the assumption that $P$ guarantees fork sequential 
  consistency in all executions.
\end{proof}

\section{Conclusions}

When clients store their data on an untrusted server, strong
guarantees should be provided whenever the server is correct, and
forking conditions when the server is faulty. Since it was discovered
that fork-linearizability does not allow for protocols that are
wait-free in all executions where the server is correct~\cite{CSS07},
the weaker condition of fork sequential consistency was expected to be
a promising direction to remedy this shortcoming~\cite{CSS07,
  oprrei06}. In this paper we proved that this is not the case, and in
fact, fork sequential consistency suffers from the same limitation.

{\small
\bibliographystyle{abbrv}
\bibliography{fscblocking}

\begin{thebibliography}{10}

\bibitem{AW94}
H.~Attiya and J.~L. Welch.
\newblock Sequential consistency versus linearizability.
\newblock {\em ACM Transactions on Computer Systems}, 12(2):91--122, 1994.

\bibitem{CSS07}
C.~Cachin, A.~Shelat, and A.~Shraer.
\newblock Efficient fork-linearizable access to untrusted shared memory.
\newblock In {\em Proc.\ 26st ACM Symposium on Principles of Distributed
  Computing (PODC)}, pages 129--138, 2007.

\bibitem{subversion}
{Collabnet, Inc.}
\newblock {Subversion} project.
\newblock \url{http://subversion.tigris.org/}, Last accessed Apr. 2008.

\bibitem{googledocs}
{Google, Inc.}
\newblock {Google Docs}.
\newblock \url{http://docs.google.com/}, Last accessed Apr. 2008.

\bibitem{herlih91}
M.~Herlihy.
\newblock Wait-free synchronization.
\newblock {\em ACM Transactions on Programming Languages and Systems},
  11(1):124--149, Jan. 1991.

\bibitem{herwin90}
M.~P. Herlihy and J.~M. Wing.
\newblock Linearizability: A correctness condition for concurrent objects.
\newblock {\em ACM Transactions on Programming Languages and Systems},
  12(3):463--492, July 1990.

\bibitem{Lamport79}
L.~Lamport.
\newblock How to make a multiprocessor computer that correctly executes
  multiprocess programs.
\newblock {\em IEEE Transactions on Computers}, 28(9):690--691, 1979.

\bibitem{lkms04}
J.~Li, M.~Krohn, D.~Mazi{\`e}res, and D.~Shasha.
\newblock Secure untrusted data repository ({SUNDR}).
\newblock In {\em Proc.\ 6th Symp.\ on Operating Systems Design and
  Implementation (OSDI)}, pages 121--136, 2004.

\bibitem{Lynch96}
N.~A. Lynch.
\newblock {\em Distributed Algorithms}.
\newblock Morgan Kaufmann, San Francisco, 1996.

\bibitem{mazsha02}
D.~Mazi{\`e}res and D.~Shasha.
\newblock Building secure file systems out of {Byzantine} storage.
\newblock In {\em Proc.\ 21st ACM Symposium on Principles of Distributed
  Computing (PODC)}, pages 108--117, 2002.

\bibitem{oprrei06}
A.~Oprea and M.~K. Reiter.
\newblock On consistency of encrypted files.
\newblock In {\em Proc.\ 20th Intl.\ Symp.\ on Distributed Computing (DISC)},
  volume 4167 of {\em Lecture Notes in Computer Science}, pages 254--268, 2006.

\bibitem{peshla80}
M.~Pease, R.~Shostak, and L.~Lamport.
\newblock Reaching agreement in the presence of faults.
\newblock {\em Journal of the ACM}, 27(2):228--234, Apr. 1980.

\bibitem{wikipedia:dfs}
{Wikipedia}.
\newblock List of file systems, distributed file systems section.
\newblock
  \url{http://en.wikipedia.org/wiki/List_of_file_systems#Distributed_file_syst%
ems}, Last accessed Apr. 2008.

\bibitem{web20}
J.~Yang, H.~Wang, N.~GU, Y.~Liu, C.~Wang, and Q.~Zhang.
\newblock Lock-free consistency control for web 2.0 applications.
\newblock In {\em Proc.\ 17th Intl.\ Conference on World Wide Web (WWW)}, 2008.

\end{thebibliography}
}
\end{document}